\documentclass[aps,pra,twocolumn,showpacs,superscriptaddress]{revtex4-1} 

\usepackage{graphicx}  
\usepackage{dcolumn}   
\usepackage{bm}        
\usepackage{multirow}
\usepackage{color}
\usepackage{amsthm}
\usepackage{bbm}
\usepackage{comment}
\usepackage{amsmath}
\usepackage{amssymb}   
\usepackage{hyperref}
\usepackage{algorithm}
\usepackage{algpseudocode}

\newcommand{\bra}[1]{\langle #1|}
\newcommand{\ket}[1]{|#1\rangle}

\newtheorem{theorem}{Theorem}

\newtheorem{definition}[theorem]{Definition}

\newtheorem{lemma}[theorem]{Lemma}

\definecolor{orange}{rgb}{1,0.66,0}

\usepackage{xr}
\begin{document}

\title{Changing the circuit-depth complexity of measurement-based quantum computation with hypergraph states}

\author{Mariami Gachechiladze}

\address{Naturwissenschaftlich-Technische Fakult\"at,
Universit\"at Siegen, 57068 Siegen, Germany}

\author{Otfried G\"uhne}
\address{Naturwissenschaftlich-Technische Fakult\"at,
Universit\"at Siegen, 57068 Siegen, Germany}

\author{Akimasa  Miyake}

\address{Center for Quantum Information and Control, 
Department of Physics and Astronomy,
University of New Mexico, Albuquerque, NM 87131, USA}

\date{\today}

\begin{abstract}
While the circuit model of quantum computation defines its logical depth  or ``computational time'' in terms of temporal gate sequences, the measurement-based model could allow totally different temporal ordering and parallelization of logical gates. By developing techniques to analyze Pauli measurements on multi-qubit hypergraph states generated by the Controlled-Controlled-Z ($CCZ$) gates, we introduce a {\em deterministic} scheme of universal measurement-based computation. In contrast to the cluster-state scheme where the Clifford gates  are parallelizable, our scheme enjoys massive parallelization of $CCZ$ and \textit{SWAP} gates, so that the computational depth grows with the number of global applications of Hadamard gates, or, in other words, with the number of changing computational bases. A logarithmic-depth implementation of an $N$-times Controlled-Z  gate illustrates a novel trade-off between space and time complexity. 
\end{abstract}

\maketitle


\section{Introduction}
A typical way to build a computer, classical or quantum, is to first realize 
a certain set of elementary gates which can then be combined to perform 
algorithms. The set of gates is called universal if arbitrary algorithms 
can be implemented. Consequently, the concept of universality is fundamental 
in computer science. While the most common choice for the universal gate set 
in quantum circuits is a two-qubit entangling gate supplemented by certain 
single-qubit gates \cite{Nielsen2000}, the universal gate set given by the 
three-qubit Toffoli gate [or the Controlled-Controlled-Z ($CCZ$) gate for our case] and 
the one-qubit Hadamard ($H$) gate \cite{Shi2002, Aharonov2003}  
is fascinating for several reasons.

First, the Toffoli gate alone is already universal for reversible classical computation. 
Consequently, the set may give insight into fundamental questions about 
the origin of quantum computational advantage, in the sense that changing the 
bases among complementary observables (by the Hadamard gates) brings  power 
to quantum  computation \cite{Shi2003, Dawson2005, Shepherd2006, Shepherd2010, 
Montanaro2016}. Second, this gate set allows certain transversal implementations  
of fault-tolerant universal quantum computation using topological error correction codes.
Transversality means that, in order to perform gates on the encoded logical qubits, 
one can apply corresponding gates to the physical qubits in a parallel fashion, and this convenience has sparked recent 
interest on this gate set \cite{Paetznick2013, Kubica2015, Yoshida2016, Yoshida2017, 
Yoder2017, Vasmer2018}. Third, the many-body entangled states generated by the $CCZ$ 
gates  are known as hypergraph states in entanglement theory \cite{Qu2013, Rossi2013, Guehne2014, Morimae2017a, Lyons2017}. 
They found applications in quantum algorithms \cite{Rossi2014} and 
Bell inequalities \cite{Gachechiladze2016}. Furthermore, as discussed below, they were recently utilized in measurement-based 
quantum computation (MBQC) \cite{Miller2016, Miller2018}, because they overlap with renormalization-group fixed-point states of 2D symmetry-protected topological orders with global $\mathbb{Z}_2$ symmetry \cite{Chen2013}.

Motivated by these observations, we introduce a \textit{deterministic} scheme of MBQC for the gate set 
of $\{CCZ, H\}$, using multi-qubit hypergraph states. MBQC is a scheme of quantum 
computation where first a highly-entangled multi-particle state is created as a 
resource, then the computation is carried out by performing local measurements 
on the particles only  \cite{Raussendorf2001, Raussendorf2003}. Compared with the 
canonical model of MBQC using cluster states \cite{Briegel2001} generated by 
Controlled-Z ($CZ$) gates, our scheme allows to extend substantially 
several key aspects of MBQC, such as the set of parallelizable gates and the 
byproduct group to compensate randomness of measurement outcomes (see \cite{Gross2007a, Gross2007b, Gross2010} for 
previous extensions using tensor network states). Although 2D ground 
states with certain symmetry-protected topological orders (SPTO) have been shown to be universal for 
MBQC \cite{Miller2016, Miller2018,Chen2018}, our construction has a remarkable feature that it 
allows {\it deterministic} MBQC, where the layout of a simulated quantum circuit can 
be predetermined.  As a resource state, we consider hypergraph states built only from $CCZ$ unitaries. This is because (i) these states have a connection to genuine 2D SPTO, (ii) it is of fundamental interest if $CCZ$ unitaries alone are as powerful as  common hybrid resources by $CCZ$ (or so-called non-Clifford elements) and $CZ$ unitaries, and (iii) they might be experimentally relevant since it requires only one type of the entangling gate, albeit a three-body interaction (cf.\citep{Monz2009, Lanyon2009, Fedorov2012,Figgatt2017}). On a technical novelty, we derive a complex graphical rule for Pauli-$X$ basis measurements on general hypergraph states, which allows a deterministic MBQC protocol on a hypergraph state, for the first time. The rule may find independent applications in deriving entanglement witnesses \cite{Toth2005,  Gach2017}, nonlocality proofs \cite{Scanari2005, Cabello2007, Gachechiladze2016}, and verification \cite{Morimae2017a, Takeuchi2018, Zhu2018} for a large class of hypergraph states. 

As a remarkable consequence of deterministic MBQC, we demonstrate an $N$-qubit generalized Controlled-Z ($C^NZ$) gate, a key logical gate for quantum algorithms such as the unstructured database search  \cite{Molmer2011}, in a depth logarithmic 
in $N$.  Although relevant logarithmic implementations of $C^NZ$  have been studied in Refs.~\cite{Selinger2013, Maslov2016, Motzoi2017}, we highlight a trade-off between space and time complexity in MBQC, namely, reducing exponential ancilla qubits to a polynomial overhead on the expense of  increasing  time complexity  from a constant depth to a logarithmic depth, in this example.  


\begin{table}[t]
\resizebox{0.95\columnwidth}{!}{\begin{tabular}{| c | c | c |}
\hline 
&$\quad$Cluster State$\quad$ & $\quad$ Hypergraph State$\quad$\tabularnewline
\hline 
\hline 
Preparation gates & $CZ\in\mathcal{C}_{2}$ & $CCZ\in\mathcal{C}_{3}$\tabularnewline
\hline 
Measurements & Pauli $+\ \mathcal{C}_{2}$ &  Pauli \tabularnewline
\hline 
\multirow{2}{*}{Implemented gates } & $\downarrow\quad\quad\quad\downarrow$ & $\downarrow$\tabularnewline
 &$\mathcal{C}_{2}\quad\quad\quad\mathcal{C}_{3}$ & $CCZ,\, H$\tabularnewline
\hline 
Byproduct  &$\{X,Z\}$ &$\{CZ,X,Z\}$\\
\hline 
Parallelized gates &$\mathcal{C}_{2}$ &$\{CCZ^{nn},\, \mbox{\textit{SWAP}}\}$ \\
\hline 
\end{tabular}}
\caption{Features of MBQC schemes using cluster and hypergraph states. 
Our scheme with a hypergraph state implements all logical $CCZ$ and 
\textit{SWAP} gates without adaptation of measurements, leading to a 
massive parallelization of these.}
\label{table-1}
\end{table}

\section{Summary of the computational scheme} 
In MBQC, an algorithm is executed by performing local measurements 
on some entangled resource state.  Consequently, two  different physical 
resources, the entangling gates needed to prepare the state and the required
class of measurements, characterize the MBQC scheme. To provide a fine-grained 
classification, let us define the Clifford hierarchy of unitary gates 
\cite{Gottesman1999}. The unitary gates in the $k$-th level of the Clifford 
hierarchy $\mathcal{C}_k$ are defined inductively, with $\mathcal{C}_1$ 
consisting of tensor products of Pauli operators, and 
$\mathcal{C}_{k+1} = \{ U |\, \forall P \!\in\! \mathcal{C}_1, U P U^\dagger 
\in \mathcal{C}_k \}$. The gates in $\mathcal{C}_2$ form the so-called Clifford 
group, preserving the  Pauli group operators  under conjugation. They allow an 
efficient classical simulation if the initialization and read-out measurements 
are performed in the Pauli bases \cite{Gottesman1998}.

There are three relevant aspects in the complexity of MBQC.
First: the adaptation of measurement bases, namely whether 
the choice of some measurement bases depends on the results of previous measurements. Second: the notion of parallelism and 
logical depth (cf. \cite{Moore2001, Terhal2004}) in terms of the ordering of measurements. Third: due to intrinsic randomness in the measurement outcomes, there are byproduct operators sometimes to be corrected. In the canonical scheme of MBQC using the cluster state,  Pauli measurements implement Clifford gates in $\mathcal{C}_2$ without  adaptation of measurement bases, so these 
gates are parallelized.  As Clifford gates are not universal, more general measurements in the $X$-$Y$-plane of the Bloch sphere must be performed to generate unitaries in 
$\mathcal{C}_3$. The byproduct group is generated by the Pauli operators $X$ and $Z$.

\begin{figure}[t]
\centering
\includegraphics[width=1\columnwidth]{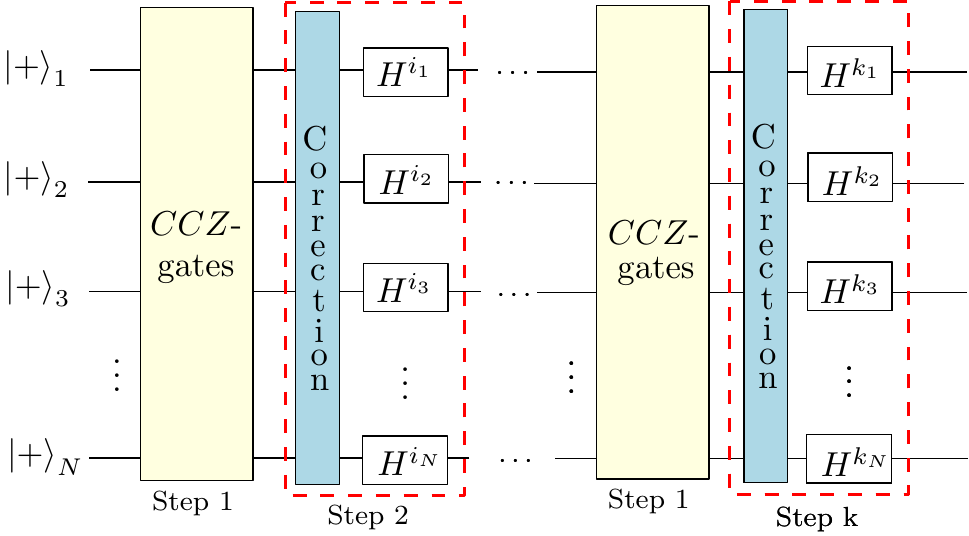}
\caption{Any quantum computation can be described as alternative applications of 
logical $CCZ$ and Hadamard gates. Our MBQC scheme allows a parallelization of
all logical $CCZ$ (namely,  $CCZ^{nn}$ and \textit{SWAP}) gates  and each Hadamard layer increments computational depth, as it  requires adaptation of 
measurement bases to correct prior byproducts.}
\label{fig:Circuit_Steps}
\end{figure}

Our scheme, however, has several key differences summarized in Table \ref{table-1}.
Our state is prepared using  $CCZ$ gates ($CCZ\in \mathcal{C}_3$), but  Pauli measurements alone are  sufficient for universal computation. We choose $\{CCZ, H\}$ to be the logical gate set for universal computation. Indeed, we can implement all logical $CCZ$ gates at {\em arbitrary} distance in parallel, by showing  that nearest neighbor $CCZ$ gates  ($CCZ^{nn}$) and \textit{SWAP} gates are applicable without adaptation. Our implementation generates the group of byproduct operators $\{CZ,X,Z\}$, 
which differs from the standard byproduct group. Since we need Hadamard gates to 
achieve universality and our byproduct group is not closed under the conjugation 
with the Hadamard gate,  we need to correct all $CZ$ byproducts before the Hadamard 
gates. Thus, the logical depth grows according to the number of global applications 
of Hadamard gates, effectively changing the computational bases (see 
Fig.~\ref{fig:Circuit_Steps}).

\section{ Hypergraph states and novel measurement rules}
Hypergraph states are generalizations of multi-qubit graph states. A hypergraph state 
corresponds to a hypergraph $H=(V,E)$, where $V$ is a set of vertices (corresponding
to the qubits) and $E$ is a set of hyperedges, which may connect more than two vertices 
(see Fig.~\ref{fig:denote} for an example). The hyperedges correspond to interactions
required for the generation of the state, as the state is defined as
\begin{equation}
\ket{H}=\prod_{e\in E}C_e\ket{+}^{\otimes |V|},
\end{equation}
where  the $C_e$'s are  generalized $CZ$ gates, 
$C_e=\mathbbm{1}-2\ket{1\dots 1}\bra{1\dots 1}$ 
acting on the Hilbert space associated to  $|e|$ 
qubits and $\ket{+}$ is a single-qubit eigenstate of the 
Pauli-$X$ observable. Hypergraph states  created by only 
three-qubit $CCZ$ gates are called three-uniform. 

In MBQC protocols  $CZ$ unitaries guarantee information flow via perfect teleportation \cite{Briegel2001, Raussendorf2003}.  Obtaining $CZ$ gates with an unit probability from  three-uniform hypergraph states has been  a challenge as  Pauli-$Z$ measurements always give $CZ$ gates  probabilistically. Therefore, only probabilistic or hybrid (where $CCZ$ and $CZ$ gates are available on demand)  scenarios have been considered in the literature \cite{Vasmer2018, Miller2016, Chen2018}. However, using a novel non-trivial Pauli-$X$ measurement rule on  three-uniform hypergraph states, we  achieve deterministic teleportation via projecting on $CZ$ gates with unit  probability. 

Note that Pauli-$X$ measurement on a graph state always projects onto a graph state, up to local unitary transformations \cite{Hein2006}. For hypergraph states, only  Pauli-$Z$ measurement rule is known \cite{Guehne2014}, while Pauli-$X$ measurements 
lead, in general, out of the hypergraph state space. In the Appendix A,  we give a sufficient criterion and a rule  for Pauli-$X$ measurements  to map hypergraph states to hypergraph states.  This rule for general hypergraph states  entirely captures the known graph state case. It can be derived by the well-known local complementation rule generalized for hypergraph states \cite{Gach2017}. Here we only give couple of examples needed later for MBQC protocol (See Appendix for more).

 For ease of notation, we draw a box instead of three vertices $V=\{1,2,3\}$ and connect it with an 
edge to another vertex $k \, (\geq 4)$  [see Fig.~\ref{fig:denote} (a)], if every two out of those three vertices are in a three-qubit hyperedge with the vertex $k$.   In addition, we say that a box is measured in the $\mathcal{M}$-basis if all three qubits $\{1,2,3\}$ 
are measured in the  $\mathcal{M}$-basis [see Fig.~\ref{fig:denote} (b), where 
$\mathcal{M}=X$]. The main two examples of measurement rules are presented in Fig. \ref{fig:outcome1} (a) and (b), where the 
post-measurement states are graph states with unit probability. By direct 
inspection one can  check that there are only two possible local Clifford 
equivalent post-measurement states when  $\mathcal{M}=X$.

\begin{figure}[t]
\includegraphics[width=1.0\columnwidth]{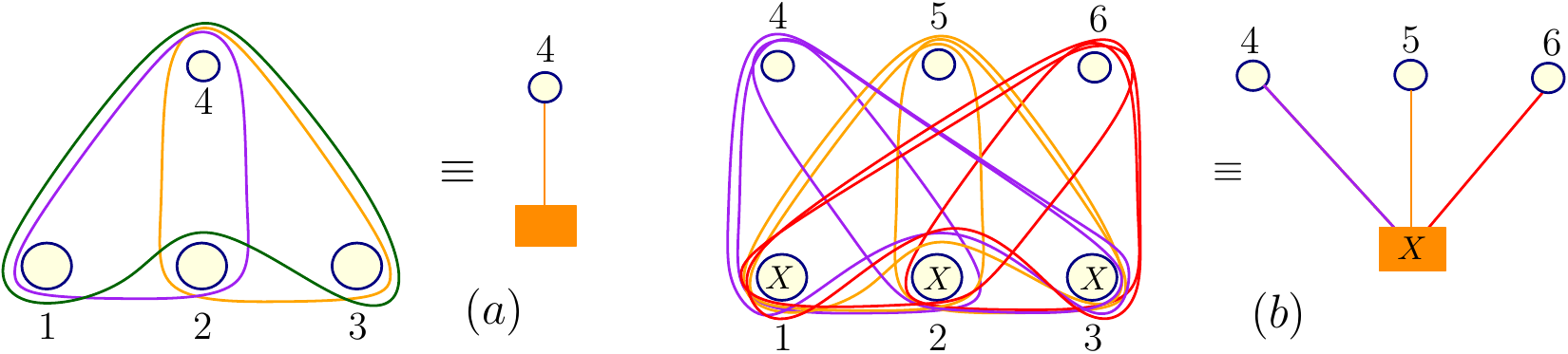} 
\caption{(a) Denoting the four-qubit hypergraph state  with hyperedges 
$E= \{\{1,2,4\},\{2,3,4\},\{1,3,4\}\}$ with the vertex and the box.  
(b)  Pauli-$X$ measurements on vertices $1,2,3$  by Pauli-$X$ measurement 
on the box.}
\label{fig:denote}
\end{figure}

\begin{figure}[t]
\includegraphics [width=0.95\columnwidth]{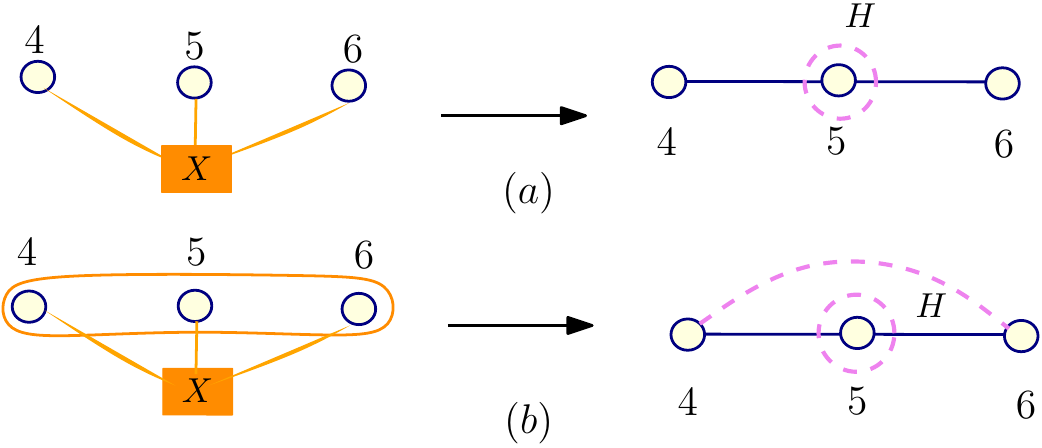}
\caption{
Pauli-$X$-measurements on the given hypergraph states result in graph states, 
with a Hadamard gate applied to its vertex $5$. All dashed lines (depicting 
byproducts)  appear additionally if the product of measurement outcomes on 
vertices $1,2,3$ is   $-1$. 
(a)  Pauli-$Z$ byproduct. 
(b)  Pauli-$Z$ and $CZ$ byproducts.}
\label{fig:outcome1}
\end{figure}

\section{ Universal resource state and  MBQC scheme}

\begin{theorem}
Based on the hypergraph state of Fig.~\ref{fig:resource}~(a), we propose MBQC with the following features:
(i) it is universal using only Pauli measurements,
(ii) it is deterministic, 
(iii) it allows parallel implementations of all logical $CCZ$ and \textit{SWAP} gates, among the universal gate set by $CCZ$, \textit{SWAP}, and Hadamard gates, and
(iv) its computational logical depth is the number of global layers of logical Hadamard gates. 
\label{theorem1}
\end{theorem}


\begin{figure}[t]
\centering
\includegraphics[width=1.0\columnwidth]{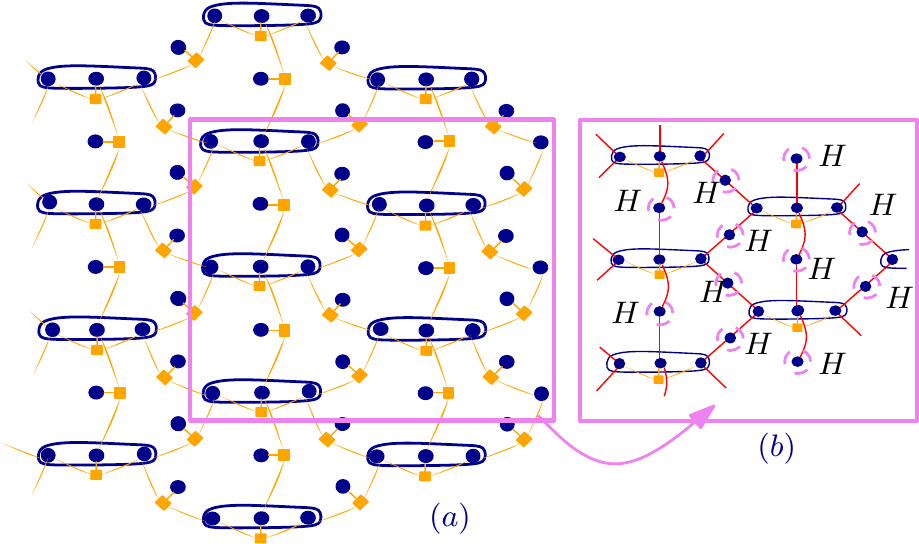}
\caption{(a) The universal resource state composed of elements 
on Fig. \ref{fig:outcome1} (a)  and (b). (b)   Resource
state obtained after measuring all boxes in Pauli-$X$ bases, except  
the ones attached to three qubits surrounded by a hyperedge. All 
dashed circles represent Pauli-$Z$ byproducts. }  
\label{fig:resource}
\end{figure}

We discuss the points in Theorem~\ref{theorem1} individually: \textit{(i)~Universality with Pauli measurements  only:}   
For the universal gate set we choose $CCZ$ and Hadamard gates.  We realize the $CCZ$ gate on arbitrary qubits in two steps: a nearest neighbor $CCZ$ gate ($CCZ^{nn}$) and a \textit{SWAP} gate, swapping an order of inputs. Here we assume that information flows from the bottom to the top.

\begin{figure}[b]
\centering
\includegraphics[scale=0.9]{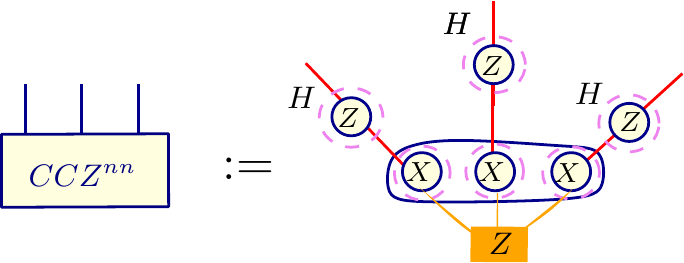}
\caption{A nearest-neighbor $CCZ$ gate is implemented up to $\{Z,CZ\}$ byproducts.
 See the Appendix B for details.}
\label{fig:gadgets}
\end{figure}

As a first step we measure almost all boxes in Pauli-$X$ basis, except the ones attached to the horizontal three vertices surrounded by a hyperedge $CCZ$. As a result we get graph  edges connecting different parts of the new state, see the transition from  Fig.~\ref{fig:resource} (a) to (b). Getting these graph edges is a crucial step, since it is partially responsible for \textit{(ii)~determinism of the protocol.} We use the resource in Fig.~\ref{fig:gadgets}  to  implement the  $CCZ^{nn}$ gate.  For $CCZ^{nn}$ gate implementation we have to secure independently three inputs and three outputs for $CCZ$ hyperedge in a hypergraph state to be used as a {\it logical} $CCZ$ gate. The box is measured in the Pauli-$Z$ basis  and just gets removed. The three vertices to which the box was attached to are  still surrounded by a hyperedge $CCZ$ up to Pauli-$Z$ byproducts.  These three qubits are connected to the rest of the state with the graph edges, and 
performing measurements as shown on Fig.~\ref{fig:gadgets} teleports the $CCZ$ gate 
to the output qubits (up to  $\{CZ, Z\}$ byproducts).  See  
Fig.~9 in the Appendix B for the explicit derivations.  

Now we need a \textit{SWAP} and a Hadamard ($H$) gate both contained in $\mathcal{C}_2$. Since some graph states  can directly implement  Clifford gates with  Pauli measurements only,  we first  get rid of all unnecessary $CCZ$ hyperedges from the resource state  by measuring all remaining boxes in 
Fig.~\ref{fig:resource}~(b) in Pauli-$X$ bases resulting to the state in 
Fig.~\ref{fig:Tohexagonal}~(b) (the full Pauli-$X$ measurement rule is needed  for the derivation) and looking at the bigger fragment, we get a graph as in Fig.~\ref{fig:Tohexagonal} (c).  The main idea here is to get rid of all the vertices which might be included in edges corresponding to byproduct $CZ$'s. Then, we make Pauli-$Z$ measurements (qubits to which an $H$ is applied, we measure 
in the Pauli-$X$ basis) on coloured vertices.  As a result, we  project to a hexagonal lattice deterministically. This construction is the final step also responsible for  \textit{(ii)~Determinism of the protocol.}  The hexagonal lattice can implement any Clifford gate in parallel  up to $\{X,Z\}$ byproducts using Pauli measurements only \cite{Nest2006}, and therefore,  we can implement a \textit{SWAP} gate. \textit{(iii)~Parallelization:} The \textit{SWAP} and $CCZ^{nn}$  gates together give a $CCZ$ gate over arbitrary distance, up to  $\{CZ,X,Z\}$ byproducts without adaptivity. 

\begin{figure}
\includegraphics[width=1.0\columnwidth]{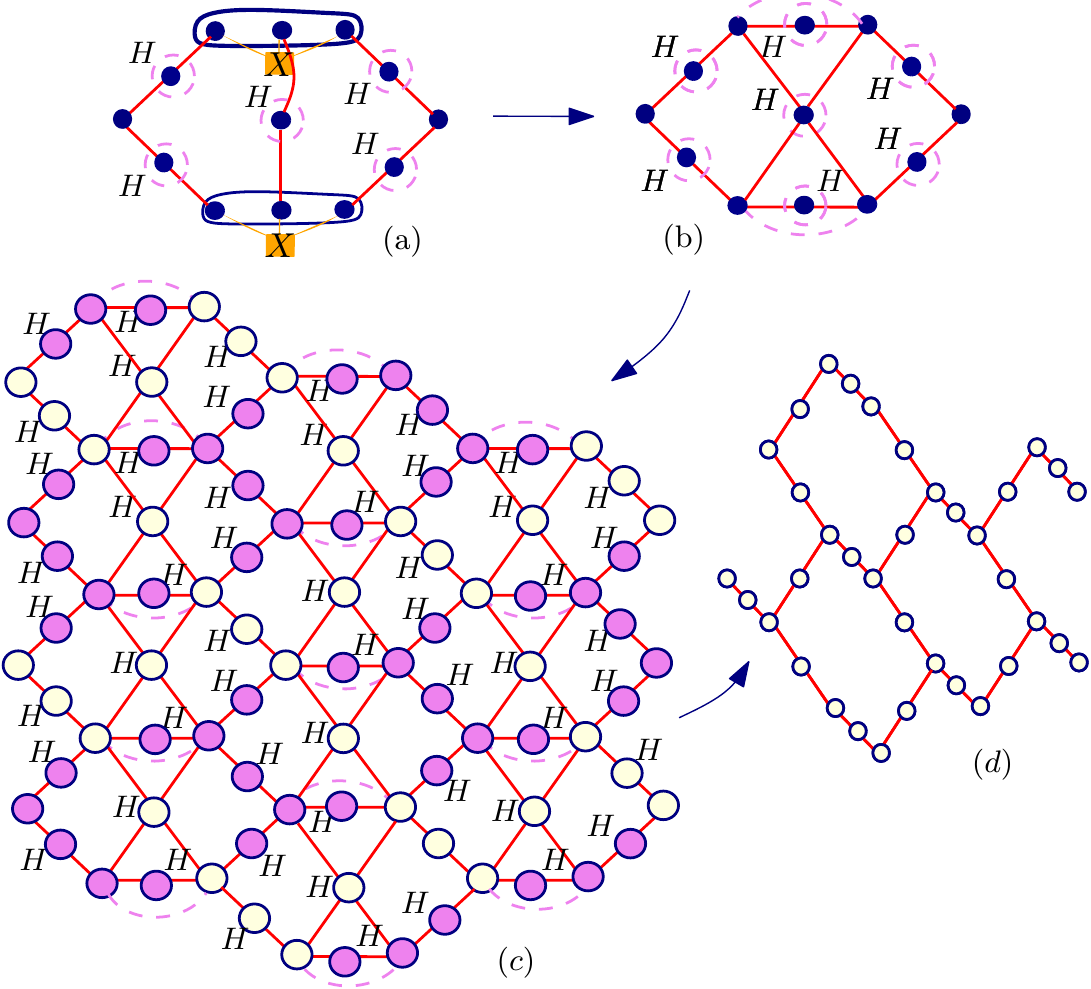}
\caption{A \textit{deterministic} graph state to implement \textit{SWAP}, H gates, and  correction steps.  (a)  Gets rid of hyperedges entirely  and projects on the graph state with Pauli-$Z$ and $CZ$ byproducts depicted 
by dashed lines in (b). The hexagonal lattice (d) is  obtained deterministically 
after measuring colored vertices in suitable Pauli bases on (c).}
\label{fig:Tohexagonal}
\end{figure}

\textit{(iv) Logical depth:} Finally,  after every $CCZ$ gate layer, we need to implement the Hadamard layer, 
which is straightforward \cite{Raussendorf2003}. However, since  $CZ$ byproducts cannot be fed-forward through Hadamard gates, we need to correct all $CZ$'s.  We can  again use the hexagonal lattice to perform 
the correction step, however, the $(k-1)$-th correction  step as enumerated in 
Fig.~\ref{fig:Circuit_Steps} itself introduces $\{X,Z\}$ byproducts which due to 
the commutation relation, $C_{abc}X_a=X_aC_{abc}C_{bc}$,  introduces new 
$CZ$ byproducts  before the $k$-th correction step. Consequently, the measurement 
results  during the $(k-1)$-th correction must be taken into account to correct 
all $CZ$ byproducts before the $k$-th  correction step. 
To sum up, we can parallelize all $CCZ$ gates, but we need to increment the  
circuit depth for each  Hadamard layer in order to correct all $CZ$ byproducts {\it adaptively}.  

\section{ Applications of parallelization}   
 We demonstrate that the parallelization in our MBQC protocol may find several practical applications, by considering an example of an $N$-times Controlled-Z ($C^NZ$) gate. Its implementation has been known either (i) in an $O(\log N)$ non-Clifford $T$ depth with $(8N-17)$ logical T-gates,  $(10N-22)$ Clifford gates  and $\left\lceil (N-3)/2\right\rceil $ ancillae \cite{Selinger2013, Maslov2016}, or (ii) in a constant depth (or constant rounds of adaptive measurements) albeit with $O(\exp N)$ $CZ$ gates in the cluster-state MBQC model and $O(\exp N)$ ancillae \cite{Raussendorf2003}. In our approach, a decomposition of the $C^NZ$ gate by $CCZ$ gates and a few number of Hadamard layers is desired.

\begin{theorem}
An $N$-times Controlled-Z ($C^NZ$) gate is feasible in an $O(\log N)$ logical depth of the Hadamard layers (or ``Hadamard'' depth), using a polynomial spatial overhead in $N$, namely  $(2N-6)$  logical Hadamard gates,  $(2N-5)$  $CCZ$ gates and  $(N-3)$ ancillae,  where $N=3\cdot 2^r$ for a positive integer $r$. 
\label{obs:CZ}
\end{theorem}

\begin{figure}
\centering
\includegraphics[width=1.0\columnwidth]{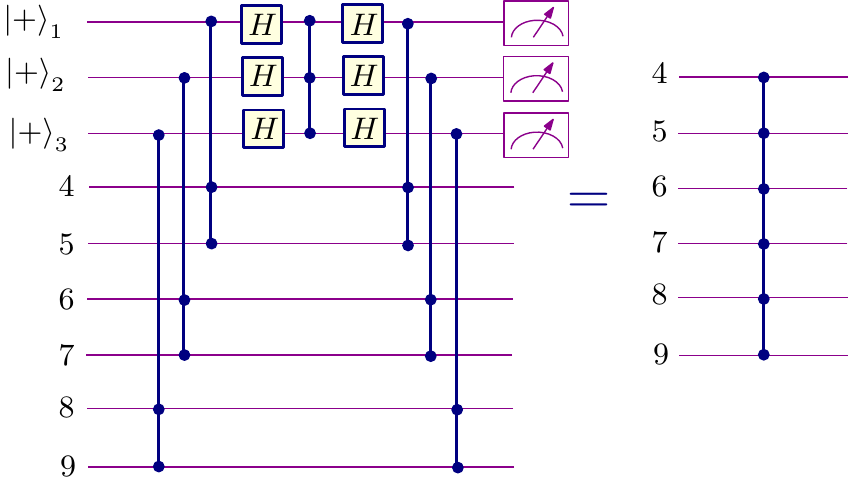}
\caption{The circuit identity to  create a $C^6Z$ gate using  $CCZ^{nn}$, \textit{SWAP},  and Hadamard gates.}
\label{fig:C6Z}
\end{figure}

The detailed derivation of the gate identity and the resource count is given in Fig.~\ref{fig:C6Z} and the Appendix C.  Note that the $T$ depth \cite{Selinger2013, Maslov2016} of \textit{(i)} and the Hadamard depth in Theorem~\ref{obs:CZ} are both logarithmic in this example. However, while the former counts  the depth of gates in $\mathcal{C}_3$ as a rough estimate in fault-tolerant quantum computation, the latter gives the depth according to the count in $\mathcal{C}_2$. Note that the $T$ depth in general is not the actual circuit depth of a unitary-gate sequence as it involves other non-commuting gates in $\mathcal{C}_2$. Our Hadamard depth, however, is indeed the actual logical depth of computation. Comparing \textit{(ii)} with our Theorem~\ref{obs:CZ}, the depth can be made constant in $N$ on a cluster state, if the number of physical qubits used in the MBQC protocol is allowed to be $2^N-1$  \cite{Raussendorf2003}.  Note that our construction in Theorem~\ref{obs:CZ} can be adapted on a cluster state by creating $CCZ$ in a constant depth and applying Theorem~\ref{obs:CZ}, so that the depth can be logarithmic in $N$ with a polynomial number of physical qubits (see Appendix C).  Therefore, Theorem~\ref{obs:CZ} demonstrates a general trade-off between space and time complexity required for quantum algorithms, from the perspective of MBQC.

\section{ Summary and outlook} 
We introduced a deterministic scheme of MBQC for the gate set of $CCZ$ and 
Hadamard gates, using a three-uniform hypergraph state and Pauli measurements. 
It enables us to parallelize massively all long-range $CCZ$ gates and the computational depth 
grows as we change computational bases. To take a broader perspective, one can 
define the Fourier hierarchy (FH) \cite{Shi2003, Shepherd2006, Shepherd2010} 
in terms of the number of the global change of the bases (namely, 
the globally parallel application of $H$ gates).  
Notably, classical polynomial-time computation, called the complexity class P, 
belongs to the 0th-level of FH. Since it is known that several 
important quantum algorithms, such as Kitaev's phase estimation, 
belong to the 2nd-level of FH (which requires only two layers of 
global $H$ gates) \cite{Shepherd2010}, it would be interesting 
to explore the implementations of low-level FH algorithms in our 
formulation. The recent major result by Bravyi et al. \cite{Bravyi2018} 
which proved quantum exponential advantage in the 2D Hidden Linear 
Function problem using a shallow circuit in the 2nd-level of FH 
is really encouraging towards this research direction 
(see e.g., \cite{Bremner2016, Bremner2017, Morimae2017b}). 

\section{Acknowledgment} 
We would like to thank D. Orsucci and J. Miller for scientific discussions, 
and D.-S. Wang for introducing to us his deterministic 
teleportation protocol using the five-qubit three-uniform hypergraph state 
in  Fig.~8 of the Appendix C. M.G. would like to thank 
A. Miyake,  O. G\"uhne, and G. Cordova for making her visit in Albuquerque 
possible, and the entire  CQuIC group at UNM for  hospitality.  M.G. and O.G 
acknowledge financial support from the DFG and the ERC (Consolidator Grant 
683107/TempoQ).  M.G. acknowledges funding from the Gesellschaft der Freunde 
und F\"orderer der Universit\"at Siegen.  A.M. is supported in part by 
National Science Foundation grants PHY-1521016 and PHY-1620651.

\appendix

\onecolumngrid
\section{\label{sec:level1}The Pauli-$X$ Measurement Rule for Hypergraph States }

In this section we derive the  Pauli-$X$ measurement rule for hypergraph states. 
We  give a sufficient criterion for the Pauli-$X$ measurements on a hypergraph 
state to project on a post-measurement state which is local unitary equivalent 
to some other hypergraph state. This criteria entirely captures the rule for graph states. 
For formulating the criterion we introduce a term which can be regarded as a generalization 
of the term  neighbourhood known in graph theory. We call it \textit{adjacency} of a vertex 
$a \in V$ and denote it by $\mathcal{A}(a) = \{e -\{a\}|e\in E \mbox{ with } a \in e\}$.

\begin{definition}
Given a hypergraph state $\ket{H}$ corresponding to a hypergraph $H=(V,E)$. If we write this hypergraph state  as follows,
\begin{equation}
\ket{H}=\frac{1}{\sqrt{2}}\ket{0}_a\ket{H_0}+\frac{1}{\sqrt{2}}\ket{1}_a\ket{H_1},
\end{equation}
we say that hypergraph state is expanded over a vertex $a\in V$.  By definition  $\ket{H_0}$ and $\ket{H_1}$ are also hypergraph states respectively corresponding to  hypergraphs $H_0$ and $H_1$ with hyperedges  $E_0=\{e\in E|a\notin e\}$ and  $E_1= E_0 \cup \mathcal{A}(a)$. If we choose a subset of vertices $V_x\subset V$ instead of a vertex $a$, we say that hypergraph state is expanded over a set of vertices $V_x\subset V$ and expansion is done iteratively for every vertex in $V_x$. 
\end{definition}
For example, if we want to expand the hypergraph state $\ket{H}$ over vertices $a$ and $b$, we first expand it over $a$ and then we expand  hypergraphs $\ket{H_0}$ and $\ket{H_1}$ separately over $b$  resulting in
\begin{equation}
\ket{H}=\frac{1}{\sqrt{2}}\ket{0}_a\bigg(\ket{0}_b\ket{H_{00}}+\ket{1}_b\ket{H_{01}}\bigg)+\frac{1}{\sqrt{2}}\ket{1}_a\bigg(\ket{0}_b\ket{H_{10}}+\ket{1}_b\ket{H_{11}}\bigg).
\end{equation}

If the vertex $a$ is measured in computational basis,  the post-measurement state is  a hypergraph state $\ket{H_0}$  for the  outcome $0$ or $\ket{H_1}$ for  the outcome $1$. However, if measured in Pauli-$X$ basis, then the post-measurement state is $\propto (\ket{H_0}\pm \ket{H_1})$ and is not always local unitary equivalent to a hypergraph state.  To check if for a given hypergraph state measuring a vertex $a$ or a set of  vertices $V_a$ in Pauli-$X$ basis gives a state local unitary equivalent to a hypergraph state, one can expand an original hypergraph state over a vertex  $a$ or a set $V_a$ and  check if all possible equally weighted superposition of expanded hypergraph states  gives some other hypergraph state or a state which is local unitary equivalent to a hypergraph state.
  
Let us consider particular cases of hypergraph states $\ket{H}$ which when expanded  over three vertices  $1,2,3$,  gives eight new  hypergraphs satisfying the following constraints  $H_{000}=H_{001}=H_{010}=H_{100}\equiv H_{\alpha}$ and $H_{111}=H_{110}=H_{101}=H_{011}\equiv H_{\beta}$. Then the expanded state can be written as follows: 
\begin{equation}\label{Eq:boxAlgebra}
\ket{H}=\frac{1}{\sqrt{8}}\bigg((\ket{000}+\ket{001}+\ket{010}+\ket{100})_{123}\otimes\ket{H_{\alpha}}+(\ket{111}+\ket{110}+\ket{101}+\ket{011})_{123}\otimes\ket{H_{\beta}}\bigg).
\end{equation}
If qubits $1,2,3$ are all measured in Pauli-$X$ bases, due to the symmetry of the first three qubits,  there are only four possible post measurement states presented in Table \ref{Table:postmeasurement}. We see from Table \ref{Table:postmeasurement} that outcome $\bra{+--}$ never occurs and outcomes $\bra{++-}$ and  $\bra{---}$ are equivalent to each other up to the global sign. Therefore, if we measure the first three qubits of the hypergraph state $\ket{H}$ as presented in Eq.~(\ref{Eq:boxAlgebra}),  there are only two possible post-measurement states and they correspond to the equally weighted superposition of two hypergraph states $\ket{H_\alpha}\pm \ket{H_\beta}$. These three qubits and their adjacencies are of our interest and in the main text they are denoted by a box. Below we  consider three examples where we measure these three qubits but we vary the hypergraphs $H_\alpha$ and $H_\beta$.

\begin{center}

\begin{table}[t]
\begin{tabular}{|c|c|c|}
\hline 
\# & Outcome & Post-measurement state\tabularnewline
\hline 
\hline 
1. & $\bra{+++}_{123}$ & $\propto(\ket{H_{\alpha}}+\ket{H_{\beta}})$\tabularnewline
\hline 
2. & $\bra{++-}_{123}$ & $\propto(\ket{H_{\alpha}}-\ket{H_{\beta}})$\tabularnewline
\hline 
3. & $\bra{+--}_{\ensuremath{123}}$ & 0\tabularnewline
\hline 
4. & $\bra{---}_{\ensuremath{123}}$ & $\propto-(\ket{H_{\alpha}}-\ket{H_{\beta}})$\tabularnewline
\hline 
\end{tabular}
\caption{All possible post-measurement states for Pauli-$X$ measurements on qubits $1,2,3$ in 
Eq.~(\ref{Eq:boxAlgebra}). Case 2 and 4 are equivalent up to a global sign. }
\label{Table:postmeasurement}
\end{table}
\end{center}

The equally weighted superposition of two hypergraph states is not always a hypergraph state again unless we choose  two hypergraphs $H_\alpha$ and $H_\beta$ specifically.  Here we give a sufficient criterion for  equally weighted superpositions of two hypergraph states being a hypergraph state up to local unitary operations and derive the graphical rule for such cases: 
 
\begin{theorem}\label{th:X_measurement_rule}
Let $H_\alpha=(V,E)$ and  $H_\beta=(V, E\cup \{a\}\cup\tilde{E})$, where $\tilde{E}$ are  hyperedges not containing a vertex $a\in V$. Then the equally weighted superpositions of two hypergraph states $\ket{H_\alpha}$ and $\ket{H_\beta}$ up to the Hadamard  gate acting on the vertex $a$, $H_a$ are still  hypergraph states  denoted by $\ket{H_{+}}$ and $\ket{H_{-}}$: 
\begin{equation}\label{Eq:AppendixoutcomePlus}
H_a\ket{H_{+}}\equiv H_a(\ket{H_\alpha}+\ket{H_\beta})\propto\prod_{e'\in E'}C_{e'} \prod_{e_a\in \mathcal{A}^\alpha(a)}\prod_{\tilde{e}\in \tilde{E}}C_{e_a\cup \tilde{e}} C_{\tilde{e}\cup a} \ket{+}^{\otimes N},
\end{equation}
\begin{equation}\label{Eq:AppendixoutcomeMinus}
H_a\ket{H_{-}}\equiv H_a(\ket{H_\alpha}-\ket{H_\beta})\propto C_a\prod_{e'\in E'}C_{e'} \prod_{e_a\in \mathcal{A}^\alpha(a)}C_{e_a}\prod_{\tilde{e}\in \tilde{E}}C_{e_a\cup \tilde{e}} C_{\tilde{e}\cup a} \ket{+}^{\otimes N}.
\end{equation}
Here $\mathcal{A}^\alpha(a)$ is the adjacency of the  vertex $a$ in hypergraph  $H_\alpha$   and $E'=\{e'|a\notin e' ,\  e'\in E\}$  and  $C_a=Z_a$.
\end{theorem}

\begin{proof}
Let us assume that $a=1$. Then we get: 
\begingroup
\allowdisplaybreaks
\begin{align}
H_1\ket{H_{+}}&=H_1(\ket{H_\alpha}+\ket{H_\beta})\\
&=H_1(\ket{H_\alpha}+Z_1 \prod_{\tilde{e}\in \tilde{E}}C_{\tilde{e}}\ket{H_\alpha})\\
& =H_1(\prod_{e\in E} C_e \bigg( \ket{+}^{\otimes N}+Z_1 \prod_{\tilde{e}\in \tilde{E}}C_{\tilde{e}} \ket{+}^{\otimes N}\bigg)\\
& =H_1\prod_{e\in E} C_eH_1 H_1 \bigg(\bigg[\ket{+}+\ket{-} \prod_{\tilde{e}\in \tilde{E}}C_{\tilde{e}}\bigg] \ket{+}^{\otimes {N-1}}\bigg)\label{Eq:Hplusline4}\\
& =H_1\prod_{e'\in E'} C_{e' }\prod_{e''\in E''} C_{e''} H_1 H_1 \bigg(\bigg[\ket{+}+\ket{-} \prod_{\tilde{e}\in \tilde{E}}C_{\tilde{e}}\bigg]\ket{+}^{\otimes {N-1}}\bigg)\label{Eq:Hplusline5}\\
\phantom{H_1\ket{H_{+}}}& =\prod_{e'\in E'} C_{e' }H_1\prod_{e''\in E''} C_{e''} H_1 H_1 \bigg(\bigg[\ket{+}+\ket{-} \prod_{\tilde{e}\in \tilde{E}}C_{\tilde{e}}\bigg] \ket{+}^{\otimes {N-1}}\bigg)\label{Eq:Hplusline6}\\
\phantom{H_1\ket{H_{+}}}
&=\prod_{e'\in E'}C_{e'} \prod_{e_1\in \mathcal{A}^\alpha(1)}CNOT_{e_1,1} \bigg(\bigg[\ket{0}+\ket{1} \prod_{\tilde{e}\in \tilde{E}}C_{\tilde{e}}\bigg]\ket{+}^{\otimes {N-1}}\bigg) \label{Eq:Hplusline7}\\
\phantom{H_1\ket{H_{+}}}&\propto\prod_{e'\in E'}C_{e'} \prod_{e_1\in \mathcal{A}^\alpha(1)}CNOT_{e_1,1} \prod_{\tilde{e}\in \tilde{E}}C_{\tilde{e}\cup 1} \ket{+}^{\otimes N} \label{Eq:Hplusline8}\\
\phantom{H_1\ket{H_{+}}}&=\prod_{e'\in E'}C_{e'} \prod_{e_1\in \mathcal{A}^\alpha(1)}\prod_{\tilde{e}\in \tilde{E}}C_{e_1\cup \tilde{e}} C_{\tilde{e}\cup 1} \ket{+}^{\otimes N} \label{Eq:Hplusline9}
\end{align}
\endgroup

In Eq.~(\ref{Eq:Hplusline4}) we  decompose  a set of hyperedges $E$ into two parts: $E'$, hyperedges which do not contain the vertex 1 and $E''$ hyperedges which contain the vertex 1.  In Eq.~(\ref{Eq:Hplusline5}) the set of hyperedges $\prod_{e'\in E'}C_{e'}$ commute with $H_1$ and going to Eq.~(\ref{Eq:Hplusline6}),  $H_1\prod_{e''\in E''}C_{e''}H_1=\prod_{e_1\in\mathcal{A}^\alpha(1)}CNOT_{{e_1},1}$, since Hadamard gate $H_1$ changes $Z_1$ to $X_1$ and, therefore, generalized Controlled-$Z$ gates become generalized $CNOT$ gates. 

 In  Eq.~(\ref{Eq:Hplusline6}), $H_1$ is applied to $\ket{\pm}$  and in  Eq.~(\ref{Eq:Hplusline7}) a new hypergraph state is obtained, which is written in an expanded form over vertex $1$. If we write this hypergraph state we get Eq.~(\ref{Eq:Hplusline8}):
 
\begin{equation}\label{Eq:Hplusline7plusHGstate}
\bigg(\bigg[\ket{0}+\ket{1} \prod_{\tilde{e}\in \tilde{E}}C_{\tilde{e}}\bigg] \ket{+}^{\otimes {N-1}}\bigg)\propto \prod_{\tilde{e}\in \tilde{E}}C_{\tilde{e}\cup 1} \ket{+}^{\otimes N}.
\end{equation}
Then generalized $CNOT$ gates are applied to a new hypergraph state in Eq.~(\ref{Eq:Hplusline8}). The action of generalized $CNOT$ gate was described in Ref. \cite{Gach2017} as follows:  Applying the generalized $CNOT_{Ct}$ gate to a hypergraph state,  where a set of control qubits $C$ controls the target qubit $t$,   introduces or deletes the set of edges  $E_t=\{e_t\cup C|e_t\in\mathcal{A}(t)\}$.  

In Eq.~(\ref{Eq:Hplusline8})  the generalized $CNOT$ gate is applied to  the hypergraph state which corresponds to the hypergraph $(V, \{\tilde{e}\cup \{1\}|\tilde{e}\in \tilde{E} \})$. The target qubit in the generalized $CNOT$ gate  is  the vertex $1$ and its adjacency is, therefore,   given by edge-set  $\tilde{E}$. The control qubits are presented by the edge-set $\mathcal{A}^\alpha(1)$, which correspond to the adjacency of the vertex $1$ in the hypergraph $H_\alpha$. The action of  generalized $CNOT$ gate takes the pairwise union 
of hyperedges in $\mathcal{A}^\alpha(1)$ and $\tilde{E}$ and adds or deletes new hyperedges: 
\begin{equation}
\prod_{e_1\in \mathcal{A}^\alpha(1)}\prod_{\tilde{e}\in \tilde{E}}C_{e_1\cup \tilde{e}}.
\end{equation}
Inserting these hyperedges in Eq.~(\ref{Eq:Hplusline9}), we get the final hypergraph states: 
\begin{equation}
H_1(\ket{H_{+}})\propto \prod_{e'\in E'}C_{e'} \prod_{e_1\in \mathcal{A}^\alpha(1)}\prod_{\tilde{e}\in \tilde{E}}C_{e_1\cup \tilde{e}} C_{\tilde{e}\cup 1} \ket{+}^{\otimes N}.
\end{equation}

In case of the minus superposition $H_1\ket{H_{-}}$,  the derivations are very similar to $H_1\ket{H_{+}}$  up to Eq.~(\ref{Eq:Hplusline7}): In particular, due to the minus sign   in the superposition, we get a different hypergraph state from the one in Eq.~(\ref{Eq:Hplusline7plusHGstate}):

\begin{equation}\label{Eq:Hplusline7minusHGstateexpanded}
H_1(\ket{+}-\ket{-} \prod_{\tilde{e}\in\tilde{E}}C_{\tilde{e}}) \ket{+}^{\otimes {N-1}}=(\ket{0}-\ket{1} \prod_{\tilde{e}\in \tilde{E}}C_{\tilde{e}}) \ket{+}^{\otimes {N-1}}=C_1\prod_{\tilde{e}\in \tilde{E}}C_{\tilde{e}\cup 1} \ket{+}^{\otimes N}
\end{equation}

 Now we apply generalized $CNOT$ gate to the hypergraph state in Eq.~(\ref{Eq:Hplusline7minusHGstateexpanded}) : 
 
\begin{equation} \label{Eq:Hplusline7minusHGstate}
 \prod_{e_1\in \mathcal{A}^\alpha(1)}CNOT_{e_1,1}C_1\prod_{\tilde{e}\in \tilde{E}}C_{\tilde{e}\cup 1} \ket{+}^{\otimes N}.
\end{equation}

The hypergraph state in Eq.~(\ref{Eq:Hplusline7minusHGstateexpanded}) has the additional edge $C_1$  and this means that the adjacency of the vertex $1$ in Eq.~(\ref{Eq:Hplusline7minusHGstate}) is given by the edge-set $\{\tilde{E}\cup \{\emptyset\}\} $.   The action of  generalized $CNOT$ gate takes the pairwise union 
of hyperedges in $\mathcal{A}^\alpha(1)$ and  $\{\tilde{E}\cup \{\emptyset\}\} $ and introduces new hyperedges of the form in the hypergraph: 

\begin{equation}
 \prod_{e_1\in \mathcal{A}^\alpha(1)}C_{e_1}\prod_{\tilde{e}\in \tilde{E}}C_{e_1\cup \tilde{e}}
\end{equation}

Inserting these hyperedges in the original derivations, gives us the final hypergraph state:
\begin{equation}
H_1\ket{H_{-}}\propto C_1\prod_{e'\in E'}C_{e'} \prod_{e_1\in \mathcal{A}^\alpha(1)}C_{e_1}\prod_{\tilde{e}\in \tilde{E}}C_{e_1\cup \tilde{e}} C_{\tilde{e}\cup 1} \ket{+}^{\otimes N}.
\end{equation}
\end{proof}

Given any graph state $\ket{G}$ corresponding to a connected graph $G=(V,E)$, if we expand it over any of its vertices $a\in V$, 
\begin{equation}
\ket{G}=\frac{1}{\sqrt{2}}\bigg(\ket{0}_a\ket{G_0}+\ket{1}_a\ket{G_1}\bigg),
\end{equation}  
then graphs corresponding to $\ket{G_0}$ and $\ket{G_1}$ satisfy the condition of Theorem \ref{th:X_measurement_rule} since $\ket{G_1}=\prod_{i\in \mathcal{N}(a)} Z_i \ket{G_0}$, where $\mathcal{N}(a)$ is the neighbourhood of the vertex $a$.  So, if the vertex $a$ is measured in Pauli-$X$ basis, the post measurement states are the equally weighted superpositions of $\ket{G_0}$ and $\ket{G_1}$ and therefore,   Theorem \ref{th:X_measurement_rule} gives the rules for deriving post-measurement states for both outcomes of measurement Pauli-$X$ basis.  The rules for Pauli-$X$ measurement for graph states was previously derived in Ref. \cite{Hein2006} using a different approach. 

\subsection{Pauli-X measurement rule using generalised local complementation on hypergraph states}

Here we  briefly review the post-measurement rules obtained for graph states using the graphical action called, \textit{local complementation} and then we generalize this result to hypergraph states. This gives a graphical rule for Pauli-$X$ measurements on hypergraph states.

Given a graph states $\ket{G}$, corresponding to a graph $G=(V,E)$, there are well defined graphical rules for obtaining post-measurement states after Pauli-$X$ measurement \cite{Hein2006} up to local corrections. The post measurement state after measuring a vertex $a$ in Pauli-$X$ basis is:
\begin{equation}
U^a_{x,\pm}\ket{\tau_{b_0}(\tau_a\circ \tau_{b_0}(G))-a},
\end{equation}
for any $b_0\in \mathcal{N}(a)$, where the map $\tau$ is local complementation and $U^a_{x,\pm}$ corresponds to a local unitary operation depending on the measurement outcome. The action  of local complementation on some vertex $a$ is defined as follows: If there were edges between pairs of vertices  in $\mathcal{N}(a)$,  erase the edges and if there is no edges between some of the vertices in  $\mathcal{N}(a)$, the edge is added between these pairs of vertices. Pauli-$X$ measurement on graph states can be described as the three consecutive applications of local complementations \cite{Hein2006}. 

Now we extend the rule to hypergraph states. We keep in mind the sufficient rule for  Pauli-$X$ measurements on hypergraph states to give a hypergraph state. Instead of writing down all the measured qubits, we can write the  following state, which would give exactly the same post-measurement states when the vertex $B$ is measured in Pauli-$X$ bases, as the original hypergraph vertices being measured in Pauli-$X$ bases  (here we are disregarding the probabilities for the post measurement-states):
\begin{equation}\label{Eq:BoxReplaced}
\ket{H_B}=\frac{1}{\sqrt{2}}\bigg(\ket{0}_B\ket{H_\alpha}+\ket{1}_B\ket{H_\beta}\bigg).
\end{equation} 
We have replaced the three qubits (a box) here with only one additional ancilla qubit $B$,  which from the structure of the hypergraphs $H_\alpha$ and $H_\beta$ evidently contains at least one graph edge connecting $B$ to the rest of the hypergraph.

We are now ready to formulate the result:
\begin{theorem}
 Given a hypergraph state $\ket{H_B}$  corresponding to a hypergraph $H_B=(V_B, E_B)$ as in Eq.(\ref{Eq:BoxReplaced}), then the post-measurement states of Pauli-$X$ basis measurement on the vertex $B$ is derived by three actions of generalized local complementation rule as follows: 
\begin{equation}\label{eq:localcomplementationrule}
U_{x,\pm}\ket{\tilde\tau_{a}(\tilde\tau_B\circ \tilde\tau_{a}(\ket{H_B}))-\{B\}},
\end{equation}
where $a$ and $B$ are contained in the same  graph edge,  $\{a,B\}\in E_B$  and
\begin{equation}
U_{x,+}=\mathbbm{1} \quad \quad \mbox{    and    } \quad \quad  U_{x,-}=C_{a}\prod_{e_i\in \mathcal{A}^{H_\alpha}(a)}C_{e_i}.
\end{equation}
Here $\mathcal{A}^{H_\alpha}(a)$ means that the adjacency of qubit $a$ must be taken from the hypergraph $H_\alpha$.
\end{theorem}

\begin{proof}
We first introduce the action of a generalized local complementation on vertex $B$  of an arbitrary hypergraph state $\ket{H}$:
\begin{equation}
\tilde\tau_{B}(\ket{H})=\prod_{e_i\in \mathcal{A}^{H}(B)}\prod_{e_j\in \mathcal{A}^{H}(B), i<j}C_{e_i\cup e_j }\ket{H}.
\end{equation}
Therefore,  a  pairwise union of  $\forall e_i, e_j\in \mathcal{A}^{H}(B)$, where $i< j$,  is added to the hyperedges of a hypergraph $H$ as a result of an action of a generalized  local complementation. For the physical maps and a derivation of the rule see Ref.~\cite{Gach2017}.

Now we use this rule to prove the theorem. From Theorem \ref{th:X_measurement_rule} we know that the hypergraphs have the following structure:  $H_\alpha=(V,E)$ and  $H_\beta=(V, E\cup \{a\}\cup\tilde{E})$, where $\tilde{E}$ are  hyperedges not containing a vertex $a\in V$. Therefore, the hypergraph $H_B$ indeed contains an edge $\{a,B\}$ and there is no other hyperedge in $H_B$ containing both $a$ and $B$ together.

Let us then consider the action of the first generalized local complementation $\tilde{\tau}(a)$. Note again that $a$ is only contained in the hyperedges $E \cup \{a,B\}$ :
\begin{equation}
\tilde{\tau}(a) \ket{H_B}=\tilde{\tau}(a) C_{aB}\prod_{\tilde e_i \in \tilde{E}} C_{\tilde{e_i}\cup B}\ket{+}_B\ket{H_\alpha}=    C_{aB}\prod_{e_i\in \mathcal{A}^{H_\alpha}(a)} C_{e_i\cup B } \prod_{e_j\in \mathcal{A}^{H_\alpha}(a), i<j}C_{e_i\cup e_j }  \prod_{\tilde e_i \in \tilde{E}} C_{\tilde{e_i}\cup B}\ket{+}_B\ket{H_\alpha}.
\end{equation}

Now we consider the second action, when $\tilde{\tau}(B)$ is applied to the new hypergraph. Note that the vertex $B$ is now contained in three types of hyperedges:  the every hyperedge in $\mathcal{A}^{H_\alpha}$   $\cup$ in every  hyperedge  in $\tilde{E}$ $\cup$ finally in $\{a,B\}$. We have to take a pairwise union between the types of the hyperedges and also the pairwise union within each type too:

\begin{equation}
\tilde\tau(B)\circ\tilde{\tau}(a) \ket{H_B}=C_{aB} \prod_{e_i\in \mathcal{A}^{H_\alpha}(a)} C_{e_i\cup B }   \prod_{\tilde e_i \in \tilde{E}} C_{\tilde{e}_i\cup B}C_{\tilde{e}_i\cup a}C_{\tilde{e}_i\cup e_i}\prod_{{\tilde e}_j \in \tilde{E}, i<j}  C_{\tilde{e}_i\cup \tilde{e}_j}\prod_{e'\in E'}C_{e'}\ket{+}^{\otimes |V_B|},
\end{equation}
where $E'$ are hyperedges in $H_\alpha$, which do not contain the vertex $a$.
Next step is to remove the vertex $B$ and all the hyperedges it is adjacent to:

\begin{equation}
\ket{\tau_B\circ \tau_{a}(\ket{H_B}))-\{B\}}= \prod_{e_i\in \mathcal{A}^{H_\alpha}(a)}  \prod_{\tilde e_i \in \tilde{E}} C_{\tilde{e}_i\cup a}C_{\tilde{e}_i\cup e_i}\prod_{{\tilde e}_j \in \tilde{E}, i<j}  C_{\tilde{e}_i\cup \tilde{e}_j}\prod_{e'\in E'}C_{e'}\ket{+}^{\otimes |V_B|-1}.
\end{equation}

And finally, the generalized local complementation over the vertex $a$ gives:
\begin{equation}
\ket{\tau_{a}(\tau_B\circ \tau_{a}(\ket{H_B}))-\{B\}}= \prod_{e'\in E'}C_{e'} \prod_{\tilde e_i \in \tilde{E}} C_{\tilde{e}_i\cup a}  \prod_{e_i\in \mathcal{A}^{H_\alpha}(a)}  C_{\tilde{e}_i\cup e_i}\ket{+}^{\otimes |V_B|-1}.
\end{equation}

  This expression exactly corresponds to the one in  Eq. (\ref{Eq:AppendixoutcomePlus}), the post-measurement state for the positive superposition.
  For the negative outcome we just fix the correction term $U_{x,-}$:
\begin{equation}
U_{x,-}\ket{\tau_{a}(\tau_B\circ \tau_{a}(\ket{H_B}))-\{B\}}= C_a\prod_{e'\in E'}C_{e'} \prod_{\tilde e_i \in \tilde{E}} C_{\tilde{e}_i\cup a}  \prod_{e_i\in \mathcal{A}^{H_\alpha}(a)}  C_{\tilde{e}_i\cup e_i} C_{e_i}\ket{+}^{\otimes |V_B|-1},
\end{equation}
  which exactly corresponds to the post-measurement state for negative superposition in Eq. (\ref{Eq:AppendixoutcomeMinus}).\end{proof}

\subsection{Examples of  Pauli-$X$ measurements on  hypergraph states}

Here we give examples of Pauli-$X$  measurements on hypergraph states.  In all of our examples exactly three vertices are measured in Pauli-$X$ bases. The post-measurement  states are derived by first expanding the hypergraph state over these three vertices  as shown in Eq.(\ref{Eq:boxAlgebra}), then checking if new emerging hypergraphs $H_\alpha$ and $H_\beta$ satisfy the condition of Theorem \ref{th:X_measurement_rule}. And only the final step if to apply the result of Theorem \ref{th:X_measurement_rule} to give the  post-measurement hypergraph states.

Here we only consider  three-uniform  hypergraph states and focus on  cases when post-measurement states are graph states regardless of the measurement outcomes, in general this is not the case. 

\textbf{Example 1:} The smallest  three-uniform hypergraph state which after measuring the first three qubits in Pauli-$X$ basis can deterministically project on  a Bell state is  (see  Fig. \ref{fig:5_qubit_state}) :
\begin{align}\label{eq:5_qubit_state}
\begin{split}
\ket{H_5}= C_{124}C_{125}C_{134}C_{135}C_{234}C_{235}\ket{+}^{\otimes 5}= \frac{1}{2\sqrt{2}}\bigg(&\big(\ket{000}+\ket{001}+\ket{010}+\ket{100}\big)\ket{+}^{\otimes 2}+\\
& \big(\ket{011}+\ket{110}+\ket{101}+\ket{111}\big)\ket{-}^{\otimes 2}\bigg).
\end{split}
\end{align}

The state $\ket{H_5}$ is given in the expanded form over vertices $1,2,3$ as in Eq.~(\ref{Eq:boxAlgebra}) and $\ket{H_\alpha}=\ket{+}^{\otimes 2}$ and  $\ket{H_\beta}=\ket{-}^{\otimes 2}=Z^{\otimes 2}\ket{+}^{\otimes 2}$.

We fix $a$ to be vertex $4$, $H_\alpha$ to have hyperedges $E_\alpha=\{\}$ and $H_\beta$ to have hyperedges $E_\beta=\{\{4\}\cup\tilde{E}\}$, where $\tilde{E}=\{\{5\}\}$. These two hypergraphs satisfy condition of Theorem \ref{th:X_measurement_rule}. So, measuring qubits $1,2,3$ in Pauli-$X$ basis gives two possible post-measurement hypergraph states $H_4\ket{H_{+}}\propto\ket{+}^{\otimes 2}+ \ket{-}^{\otimes 2}$ with the probability $1/5$  and $H_4\ket{H_{-}}\propto\ket{+}^{\otimes 2}+ \ket{-}^{\otimes 2}$ with the probability $4/5$. Using Theorem \ref{th:X_measurement_rule} we derive these post-measurement states:
\begin{equation}
H_4\ket{H_{+}}\propto H_4 \bigg(\ket{+}^{\otimes 2}+ \ket{-}^{\otimes 2}\bigg)\propto C_{45} \ket{+}^{\otimes 2}  \quad \mbox{ and } \quad H_4\ket{H_{-}}\propto H_4 \bigg(\ket{+}^{\otimes 2}- \ket{-}^{\otimes 2}\bigg)\propto C_{45}C_4 \ket{+}^{\otimes 2}. 
\end{equation}

\begin{figure}[!htb]
\begin{center}
\includegraphics[scale=0.75]{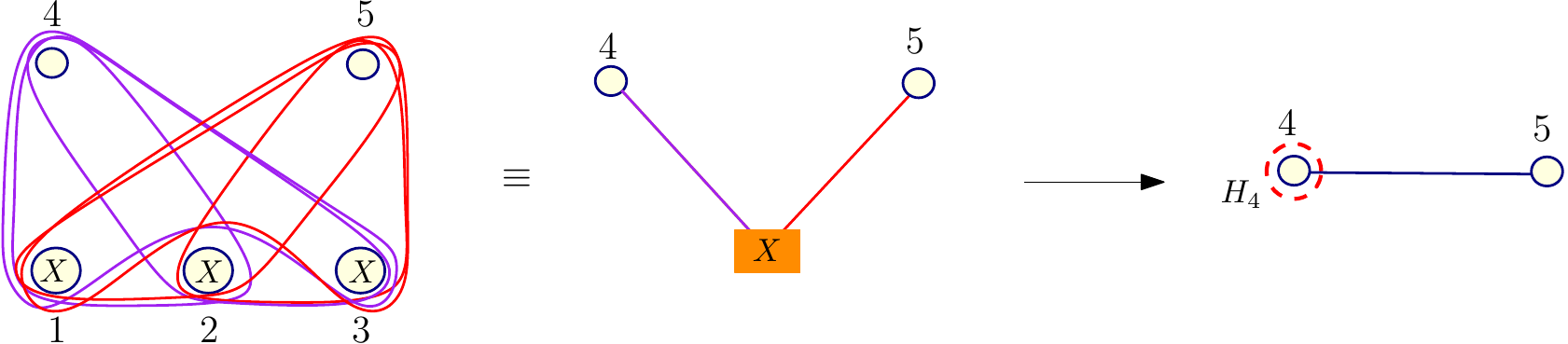}
\end{center}
\caption{The five-qubit three-uniform hypergraph state \cite{Wang2017} is the smallest hypergraph state with no usual graph edges which can project on a Bell state deterministically. It has  hyperedges $E=\{\{1,2,4\},\{1,2,5\},\{1,3,4\},\{1,3,5\},\{2,3,4\},\{2,3,5\}\}$.  The qubits $1,2,3$ are measured in $X$-basis and the post-measurement  state  is a graph state with a Hadamard correction on the vertex $4$. The  graph state is obtained with  unit probability  but up to Pauli-$Z_4$ byproduct. The probabilistic   Pauli-$Z_4$  is denoted by the dotted circle and it appears with the probability $4/5$ when the product of Pauli-$X$ measurement outcomes is $-1$.  }
\label{fig:5_qubit_state}
\end{figure}

\textbf{Example 2:} Let us consider the six-qubit hypergraph state $\ket{H_6}$ presented on Fig. \ref{fig:outcome1} (a). After measuring qubits $1,2,3$ in $X$-basis we  project on the three-qubit graph state. To see this, we write $\ket{H_6}$ directly in the expanded form  over vertices $1,2,3$:
\begin{equation}
\includegraphics[scale=1.0]{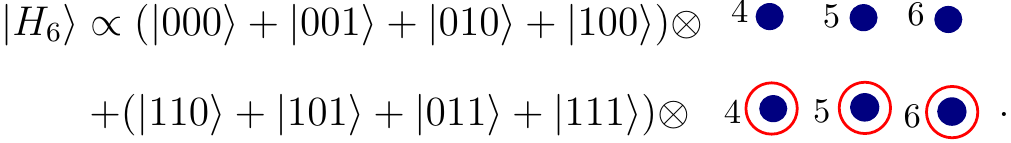}
\end{equation}

Here $H_\alpha$ has hyperedges $E_\alpha=\{\}$ and  $H_\beta$ has hyperedges $E_\beta=\{\{4\},\{5\},\{6\}\}$ and we fix to apply the Hadamard correction on the vertex $a=5$.  We can use  Theorem \ref{th:X_measurement_rule} to derive two post-measurement states upto Hadamard gate applied to the vertex $5$: 
\begin{equation}\label{Eq:minusoutcomewithoutedge}
 H_5 \ket{H_{+}}\propto C_{45}C_{56}\ket{+}^{\otimes 3} \quad \mbox{ and }\quad H_5 \ket{H_{-}}\propto C_{45}C_{56}C_5\ket{+}^{\otimes 3}.
\end{equation}

\textbf{Example 3:} Let us consider more complicated six-qubit  hypergraph state $\ket{H_6}$ presented on Fig. \ref{fig:outcome1} (b). We write this state expanded  over vertices $1,2,3$:
\begin{equation}
\includegraphics[scale=1.0]{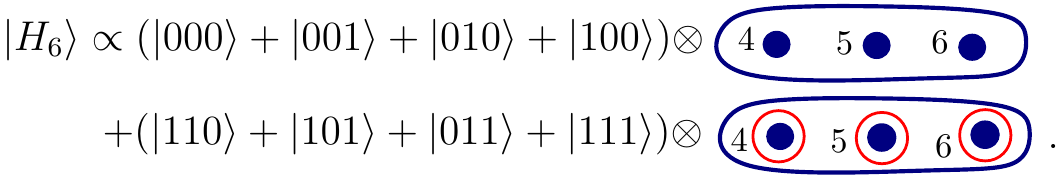}
\end{equation}
Here $H_\alpha$ has hyperedges $E_\alpha=\{\{1,2,3\}\}$ and  $H_\beta$ has hyperedges $E_\beta=\{\{1,2,3\},\{4\},\{5\},\{6\}\}$ and we fix to apply the Hadamard correction on the vertex $a=5$. We can use  Theorem \ref{th:X_measurement_rule} to derive two post-measurement states upto Hadamard gate applied to qubit $5$: 
\begin{equation}
  H_5 \ket{H_{+}}\propto C_{45}C_{56}\ket{+}^{\otimes 3} \quad \mbox{ and } \quad  H_5 \ket{H_{-}}\propto  C_{45}C_{56}C_{46}C_5\ket{+}^{\otimes 3}.
\end{equation} 

\textit{Remark.}  We can  increase the number of vertices  that we measure in Pauli-$X$ and  generalize  a notion of the box  defined  in the main text. The box that we considered up to now was corresponding to the structure of the  expanded  three vertices and was always connected to the rest of the hypergraph with three-qubit hyperedges. Now we try to extend this result to higher cardinality edges.  Let us expand a hypergraph state over $m$-qubits, where $3 \leq  m\leq N-2$ is an odd number, in the following way:
\begin{equation}
\ket{H_N}\propto  \bigg(\sum_x \ket{x} \bigg)\otimes \ket{H_\alpha}+ \bigg(\sum_y \ket{y} \bigg)\otimes \ket{H_\beta},
\end{equation}
where $x,y\in\{0,1\}^m$ and the first sum runs over all computational bases elements with the weight $w(x)\leq \left\lfloor m/2\right\rfloor$ and the second sum runs over all computational bases elements with the weight $w(y)> \left\lfloor m/2\right\rfloor$.

If all the first $m$ vertices are measured in Pauli-$X$ bases, then we again get two possible measurement outcomes $\ket{H_\alpha}\pm \ket{H_\beta}$. However, the box now can look very different from the $m=3$ case.

For simplicity let us fix $\ket{H_\alpha}=\ket{+}^{\otimes |N-m| }$ and $\ket{H_\beta}=\ket{-}^{\otimes |N-m| }$.  The smallest hyperedge the new type of a box is connected to the rest of the hypergraph has a cardinality equal to $ \left\lceil m/2\right\rceil+1$. But  in addition, for  some cases of $m$ with this construction the box will be connected to the rest of the hypergraph with different sizes of hyperedges. 

To illustrate this let us consider an example of $\ket{H_7}$, where $m=5$ and $\ket{H_\alpha}=\ket{+}^{\otimes 2}$ and $\ket{H_\alpha}=\ket{-}^{\otimes 2}$.
Then the smallest cardinality hyperedge in the hypergraph is of a size four -  the smallest weight of vector $\ket{y}$  is equal to  $ \left\lceil 5/2\right\rceil=3$ and plus  $1$.  However, these are not all the hyperedges in the hypergraph: The  vectors with the weight four are  in the second summand and they are tensored with $\ket{-}^{\otimes 2}$.   However, if we choose any four vertices among $m$,  then every three from them are connected to both vertices $m+1$ and $m+2$,  but $\binom{4}{3}=4$, which is an even number. So, the hypergraph must have additional cardinality $5$ edges. 
Similarly we have to check the weight of the last term in the sum:  $\binom{5}{3}+\binom{5}{4}=15$  is an odd number and, therefore,  there is no cardinality six edges in the hypergraph.  Therefore, similarly to $m=3$ case, we got a box containing five qubits but the box is connected to the rest of the hypergraph with four- and five- qubit hyperedges in a symmetric manner.

\section{ Implementation of $CCZ^{nn}$ gate}
Since we have chosen $\{CCZ,H\}$ to be the universal gate set, we need to show in detail how to implement these gates on our resource state. To start with, we implement $CCZ$ gates only on the nearest neighbor qubits (denote it by $CCZ^{nn}$) and therefore,  we need  \textit{SWAP} gate too. The goal is to implement all the gates deterministically. All Pauli measurements are made in one step but for simplicity  we consider them in several steps. At the step one the box is measured in Pauli-$Z$ basis. This evidently removes the box entirely and introduces Pauli-$Z$  byproducts on the vertices $1,2,3$ as presented on Fig.~\ref{fig:CCZNNsteps}.

On Fig.~\ref{fig:CCZNNsteps} at step 1  we first describe the measurements needed to get  $CCZ^{nn}$ gate using our resource state.  Now let us measure the vertex $4$ in Pauli-$Z$ basis, this effectively implements Pauli-$X$ measurement, since the Hadamard gate was applied to this vertex. We need to use Theorem \ref{th:X_measurement_rule} to derive a post-measurement state. Let us write the hypergraph state in the expanded form over the vertex $4$:  
\begin{equation}
\ket{H}=\frac{1}{\sqrt{2}}\big( \ket{0}_4 \ket{H_0}+ \ket{1}_4\ket{H_1}\big)=\frac{1}{\sqrt{2}}\big( \ket{0}_4 \ket{H_0}+ \ket{1}_4 Z_1 Z_7\ket{H_0}\big).
\end{equation}

\begin{figure}[t]
\includegraphics[scale=0.75]{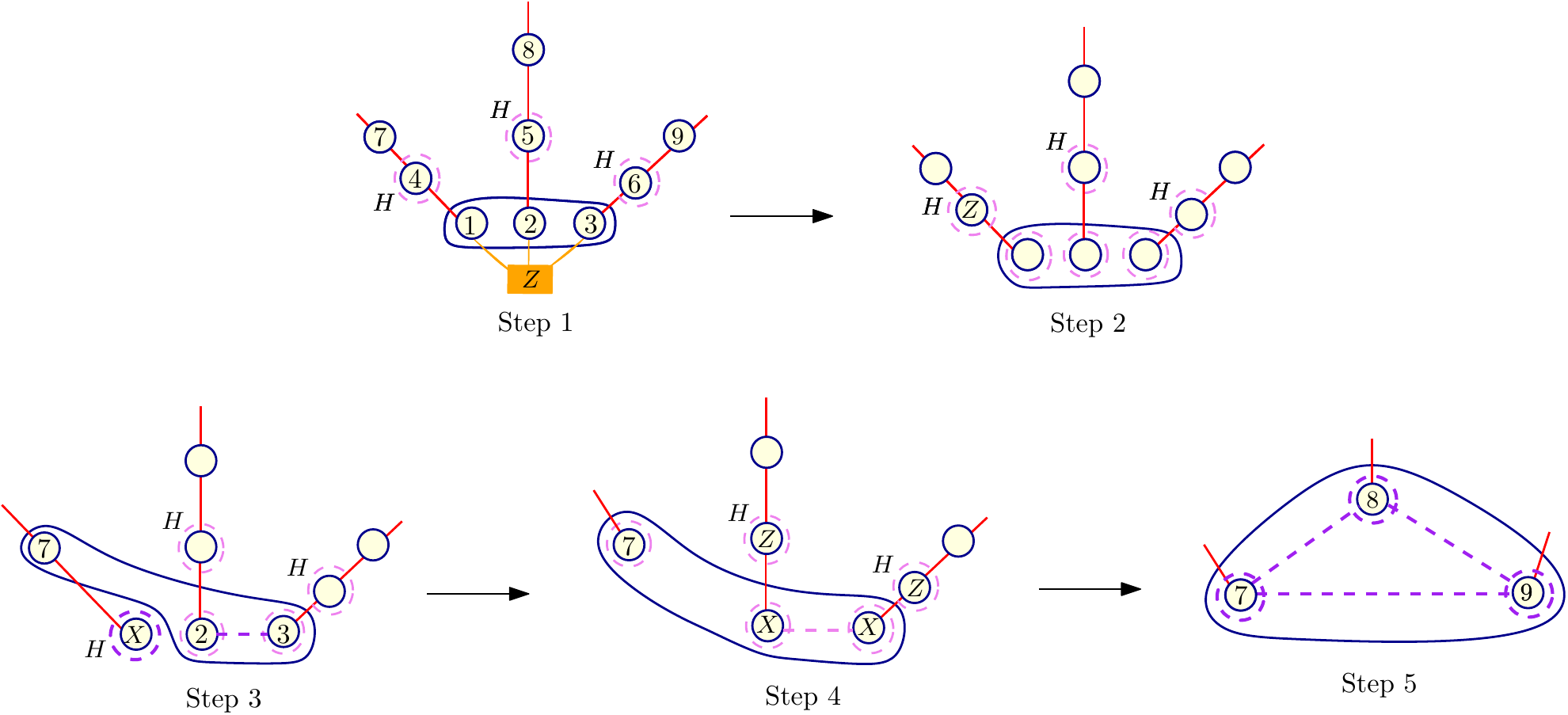}.
\caption{Implementing $CCZ^{nn}$ gate. All measurements are made simultaneously. We present them step by step to emphasize how $CZ$ byproducts come into the computational scheme.  Step 1: Pauli-$Z$measurement on the box  removes the box and introduces  Pauli-$Z$ byproducts on vertices $1,2,3$. Step 2: The vertex $4$ is measured in $X$-basis projecting on the state  at step 3. Step 3: We see the $CZ_{23}$ byproduct depending on the outcome of the measurement   on the vertex $4$. Measuring the vertex $1$ in Pauli-$Z$, projects on the hypergraph at step 4. Step 4: Repeating the measurements for vertices $2,3,5,6$ gives the state at step 5. }
\label{fig:CCZNNsteps}
\end{figure}

Then $\ket{H_0}$ and $\ket{H_1}$ satisfy the condition of Theorem \ref{th:X_measurement_rule} with a Hadamard applied on the vertex $1$  and accordingly the post-measurement state  is given on Fig.   \ref{fig:CCZNNsteps} at step 3. $CCZ$ gate is now applied to the vertices $2,3,7$.  At this step we have to point out that the post-measurement state has the edges $\{1\}$ $\{2,3\}$ for  the measurement outcome $"-1"$. Thus, this is where $CZ$ byproducts come into the computation scheme discussed in the main text.   At the step 3 the vertex $1$  is measured in Pauli-$X$ basis and since the Hadamard is applied to this qubits, this implements Pauli-$Z$ measurement instead.  Repeating this measurement pattern over as shown at the step 4 given the final state at the step 5, where $CCZ$ gate is applied to vertices $7,8,9$ upto $CZ$  and $Z$ byproducts.

\section{Discussion of the complexity}

Here we first give the proof for the gate identity from the main text. 

\begin{lemma}The following equality holds for any state $\ket{\psi}$ 
and sets $i\in e_1$ and  $i\in e_2$ :
\begin{equation}
C_{e_1}H_iC_{e_2}H_iC_{e_1}\ket{+}_i\ket{\psi}=\ket{+}_iC_{e_1\cup e_2 \backslash \{i\}}\ket{\psi}.
\end{equation}
\end{lemma}

\begin{proof}
Assume that $i=1$ and denote $e_1'\equiv e_1\backslash \{1\}$ and  $e_2'\equiv e_2\backslash \{1\}$, then $e_1\cup e_2 \backslash \{1\}=e_1'\cup e_2'$:
\begin{equation}
C_{e_1}H_iC_{e_2}H_iC_{e_1}\ket{+}_i\ket{\psi}=C_{\{1\}\cup e_1'}CNOT_{e_2',1}C_{\{1\}\cup e_1'}\ket{+}_1\ket{\psi}
\end{equation}
We can express an arbitrary multi-qubit  state $\ket{\psi}$ in Pauli-$X$ orthonormal basis $\ket{j}$:  $\ket{\psi}=\sum_j \phi_j\ket{j}$. Then each vector $\ket{+}_1\ket{j}$ is itself a hypergraph state. In Ref.~\cite{Gach2017} the action of a generalized $CNOT$ gate was described  on hypergraph states as we have already used in the previous sections:  Applying the generalized $CNOT_{Ct}$ gate to a hypergraph state,  where a set of control qubits $C$ controls the target qubit $t$,   introduces or deletes the set of edges  $E_t=\{e_t\cup C|e_t\in\mathcal{A}(t)\}$.  

\begin{figure}[t]
\begin{center}
\includegraphics[width=0.65\columnwidth]{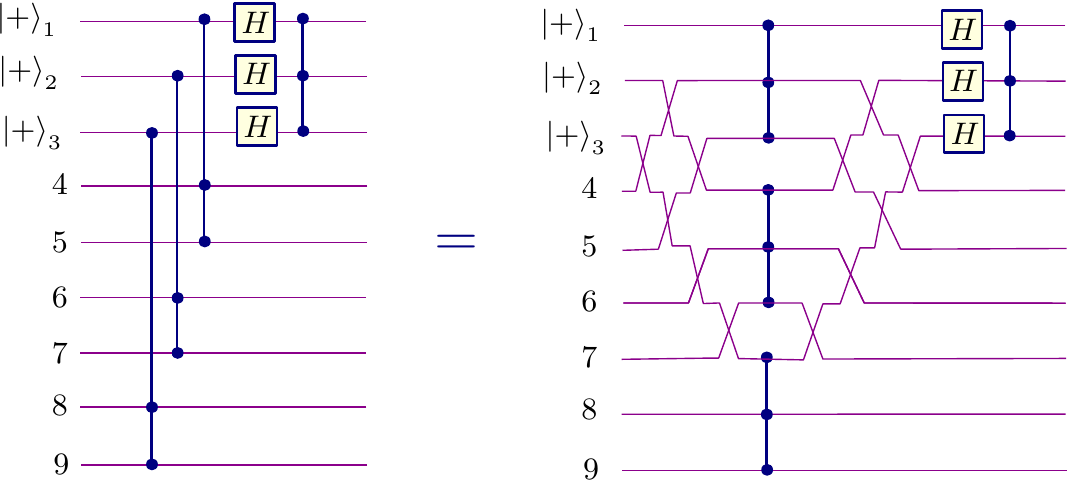}
\end{center}
\caption{(a) Circuit for $C^6Z$ gate using long-ranged $CCZ$ gates. (b) Implementation SWAP gate: Uses $9$ $CZ$ gates (brown edges) and 8 ancilla qubits.} 
\label{Fig:C6ZwithSWAPS}
\end{figure}

In our example the target qubit $t=1$ and  for each  hypergraph state $\ket{+}_1\ket{j}$  the target qubit $t=1$  is  in a single hyperedge ${C_{e_1}}$ only. Therefore from linearity follows that: 
\begin{equation}
C_{\{1\}\cup e_1'}CNOT_{e_2',1}C_{\{1\}\cup e_1'}\ket{+}_1\bigg(\sum_j\psi_j\ket{j}\bigg)=C_{\{1\}\cup e_1'}C_{e_2 \cup e_1'}C_{\{1\}\cup e_1'} \ket{+}_1\bigg(\sum_j\psi_j\ket{j}\bigg)=C_{e_2 \cup e_1'} \ket{+}_1\ket{\psi}
\end{equation}
\end{proof}

For an example let us  step-by-step consider the circuit in Fig.~\ref{fig:C6Z} in the main 
text implementing a $C^6Z$ gate:  
\begin{equation}
C_{145}H_1C_{123}H_1C_{145}\ket{+}_1\ket{+}_2\ket{+}_3\ket{\psi}_{456789}=C_{2345}\ket{+}_1\ket{+}_2\ket{+}_3\ket{\psi}_{456789}. 
\end{equation}
Applying the same identity one more time when we have a Hadamard on the second qubit (we omit the 
first qubit $\ket{+}_1$): 
\begin{equation}
C_{267}H_2 C_{2345}H_2C_{267}\ket{+}_2\ket{+}_3\ket{\psi}_{456789}
=C_{34567}\ket{+}_2\ket{+}_3\ket{\psi}_{456789}. 
\end{equation}
And finally, using the third qubit  $\ket{+}_3$ for the same identity, we get the $C^6Z$ gate 
(we again omit writing $\ket{+}_2$):
\begin{equation}
C_{389}H_3 C_{34567}H_3C_{389}\ket{+}_3\ket{\psi}_{456789}=\ket{+}_3C_{456789}\ket{\psi}_{456789}. 
\end{equation}
Measuring qubits $1,2,3$ in Pauli-$X$ bases, we get $C_{456789}=C^6Z$ gate being applied to the arbitrary state $\ket{\psi}_{456789}$.

\begin{figure}[t]
\includegraphics[width=0.5\columnwidth]{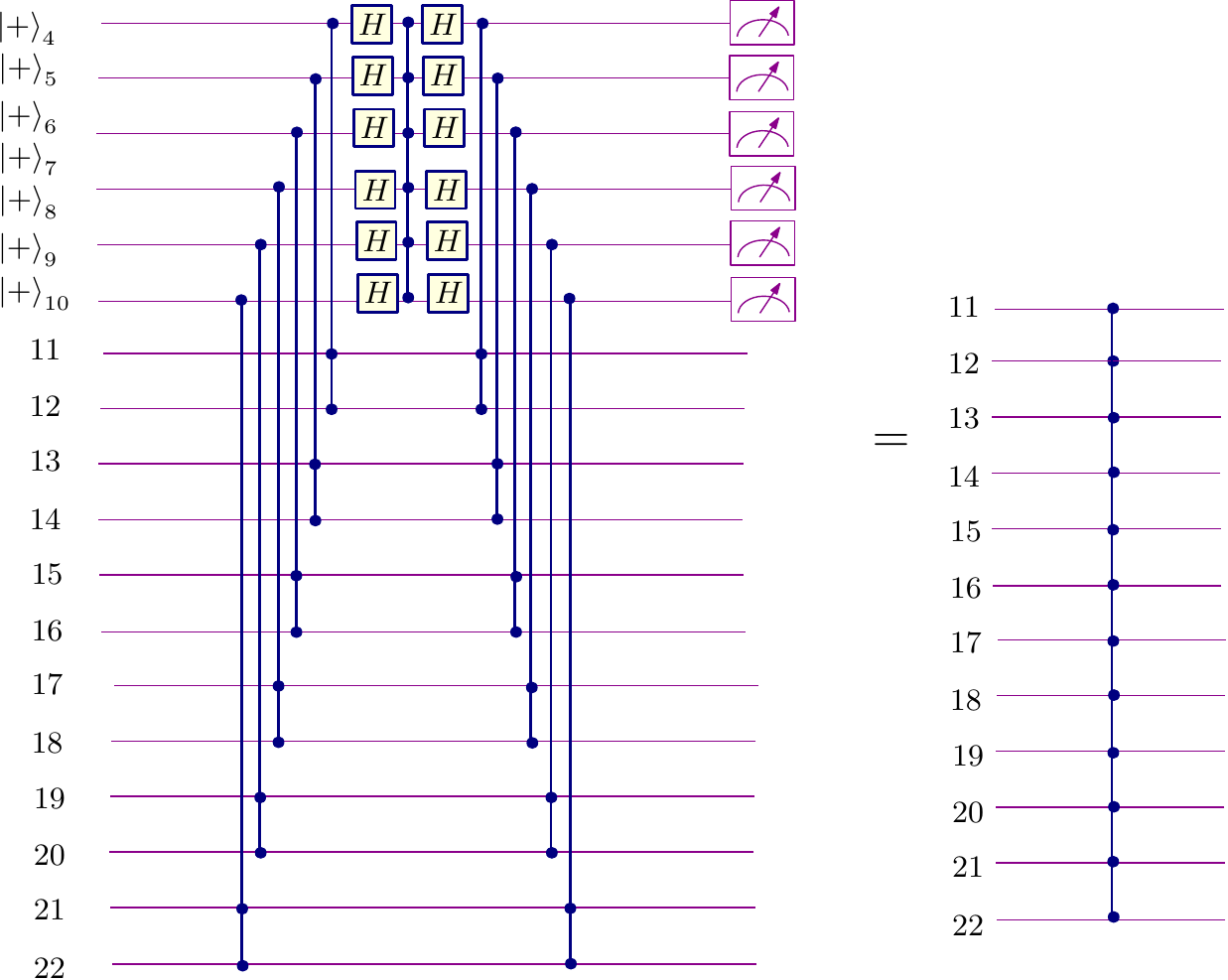}
\caption{The circuit identity implementing $C^{12}Z$ gate. Here the middle $C^6Z$ gate was created with the circuit on Fig.~\ref{fig:C6Z}. The procedure can be iterated to general  $C^NZ$ gate.}
\label{fig:C12Z}
\end{figure}

\begin{figure}[t]
\includegraphics[width=0.5\columnwidth]{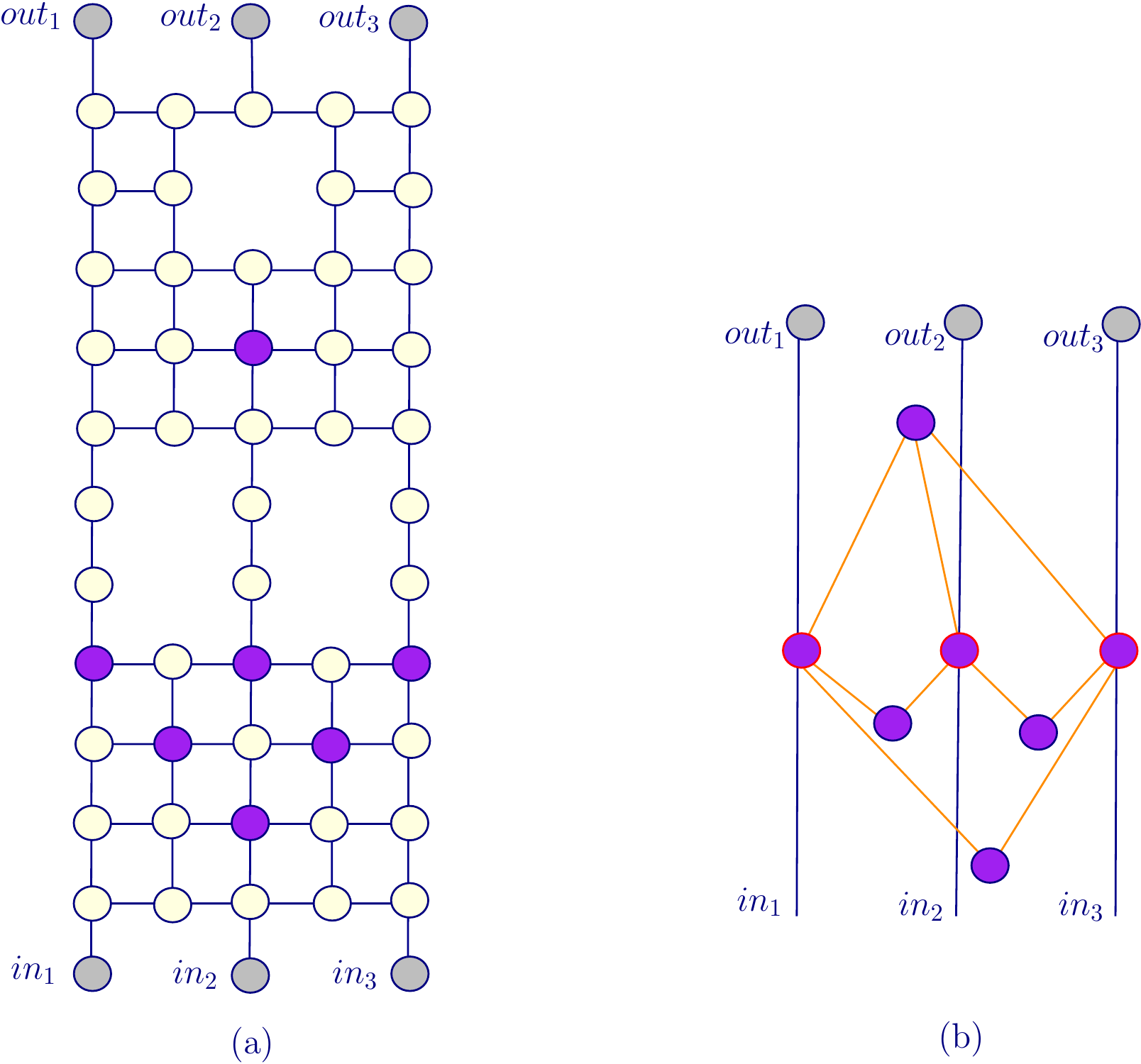}
\caption{
(a) The 55-qubit graph state given in Ref. \cite{Raussendorf2003}, which can implement three-qubit phase-gates. Six gray vertices are for input-output, seven dark purple vertices are measured in the second round of measurement, the rest is measured in Pauli-$X$ basis in the first round and the vertices which are already removed are measured in Pauli-$Z$ basis.  (b) The seven-qubit graph state obtained after Pauli measurements on (a)  capable of implementing a three-qubit phase-gates \cite{Raussendorf2003, Tsimakuridze2017}.}
\label{fig:Raussendorf}
\end{figure}

Next we count physical resource necessary to implement $C^NZ$ gate. We saw in Appendix B in Fig.~\ref{fig:CCZNNsteps}   that the minimal physical resource for $CCZ^{nn}$ gate is one physical $CCZ^{nn}$ gate, and three $CZ$ gates, represented by physical edges $\{\{1,4\}, \{2,5\},\{3,6\}\}$ and six ancilla qubits $\{1,2,3,4,5,6\}$.   The minimal physical resource for a \textit{SWAP} gate is nine  $CZ$ gates and eight ancilla qubits represented  in Fig.~\ref{Fig:C6ZwithSWAPS}  (b).  Number of total $CCZ^{nn}$ gates can be counted easily from the circuit, it also matches with number of Hadamard gates in the circuit plus one and for implementing $C^{3\cdot 2^r}$ gate  is equal to:
\begin{equation}
K_{CCZ}=3 \bigg(\sum_{k=1}^r 2^k\bigg)+1=2N-5.
\end{equation}

Here we count number of \textit{SWAP} gates needed. For $C^6Z$ we need twenty-four \textit{SWAP} gates. In general, to implement a $C^NZ$ gate with our protocol having already created  a $C^{N/2}Z$ gate,  we need $N(N-2)$ \textit{SWAP} gates.  So, in order to create $C^{N}Z$ gate we need 
to sum up \textit{SWAP} gates needed at all previous steps of iteration. If $N=3\cdot 2^r$, 
then there are totally $r=\log{(N/3)}$ iterations in our model from Observation \ref{obs:CZ}. 
To sum up, totally
\begin{equation}
K_{SWAP}=\sum_{k=1}^{r}(3\cdot 2^k)(3\cdot 2^k-2)=4N(\frac{N}{3}-1)
\end{equation}
\textit{SWAP} gates are needed. 

So, to sum up we need $K_{CCZ}=2N-5$ physical $CCZ^{nn}$  gates,  $3 K_{CCZ}+9K_{SWAP}=3(2N-5)+12N^2-36N=12N^2-30N-15$ physical $CZ$ gates, and   $6 K_{CCZ}+8K_{SWAP}=\frac{32}{3}N^2-20N-30$ physical qubits.

Next we look into the standard  protocol for creating the  $C^NZ$ gate using MBQC with cluster states.  In Ref.~\cite{Raussendorf2003} the 55-qubit cluster state is given to implement three-qubit phase-gates. Some of the vertices are missing from the cluster as they have been measured in Pauli-$Z$ basis (see Fig.~\ref{fig:Raussendorf} (a)  for the 55-qubit cluster state from Ref.~\cite{Raussendorf2003}). The gray qubits serve for input and output registers. The main idea of the protocol is to measure all the vertices displayed on Fig.~\ref{fig:Raussendorf}~(a) except the  dark purple ones  in the first round of measurements in Pauli-$X$ basis simultaneously.  Resulting post-measurement state up to Pauli-$Z$ byproducts is the seven-qubit graph state on Fig.~\ref{fig:Raussendorf} (b). Note the similarity of this graph with the graph in Fig.~5 of Ref.~\cite{Tsimakuridze2017}

We draw this graph state in the following way: The graph has  $\binom{3}{1}=3$ 
central vertices, which are connected to input-output wires , $\binom{3}{3}=1$ 
vertex adjacent to all the central qubits, and $\binom{3}{2}=3$ vertices, each 
adjacent to only two of the central vertices such that all  pairs from the central 
vertices are connected to distinct vertices. Totally, this makes 
$\binom{3}{1}+\binom{3}{2}+\binom{3}{3}=2^3-1=7$ qubits. Then, depending 
on the previous Pauli-$X$ measurement outcomes, each of these seven
qubits are measured in the two eigenbases of $U_Z(\pm\frac{\pi}{4})XU_Z(\pm\frac{\pi}{4})^\dagger$ creating a three-qubit phase-gate up to Pauli byproducts \cite{Raussendorf2003, Tsimakuridze2017}. 

If we extend this result for $C^4Z$ gate, the initial cluster state must be reduced to the graph 
state via Pauli measurements implemented in parallel. 
The structure of this graph is analogous to the one discussed for $C^3Z$ case. But now we need $\sum_{i=1}^{4}\binom{4}{i}=2^4-1$ qubits. From here one can see that to implement a $C^NZ$ gate in the standard way starting from the cluster state, one would only need to adapt measurement basis twice, which is constant for any $N$, but number of qubits one would require is $\sum_{i=1}^{N}\binom{N}{i}=2^N-1$ which is exponential  with the size of 
the gate implemented \cite{Tsimakuridze2017}. 

Let us look at the count of a physical qubits in  case our gate identity from Theorem \ref{obs:CZ}  is used on a cluster state.   As seen in Fig. \ref{fig:Raussendorf} (b) for the three-qubit phase-gate eight physical qubits are needed. The swap gate can be implemented as in Fig. \ref{Fig:C6ZwithSWAPS} (b), therefore needs  $8$ qubits. Therefore totally $8(K_{CCZ}+K_{SWAP})=\frac{32}{3}N^2-16N-40$ qubits are needed, which is polynomial in $N$. 
\mbox{ }
\vspace{0.3cm}
\mbox{ }
\twocolumngrid

\end{document}